\documentclass[hidelinks,onefignum,onetabnum]{siamart251216}

\usepackage{lipsum}
\usepackage{amsfonts}
\usepackage{graphicx}
\usepackage{epstopdf}
\usepackage{algorithmic}
\ifpdf
  \DeclareGraphicsExtensions{.eps,.pdf,.png,.jpg}
\else
  \DeclareGraphicsExtensions{.eps}
\fi

\usepackage{tcolorbox}
\tcbuselibrary{skins,breakable}
\usetikzlibrary{shadings,shadows}

\usepackage{tikz}
\usepackage{pgfplots}
\usepackage{adjustbox}

\usepackage{amstext}
\usepackage{amssymb}
\usepackage{mathrsfs}
\usepackage{subfigure}

\usetikzlibrary{arrows.meta,positioning}
\usetikzlibrary{arrows.meta,calc}
\usepackage{tkz-euclide}

\usepackage{enumitem}
\setlist[enumerate]{leftmargin=.5in}
\setlist[itemize]{leftmargin=.5in}

\newsiamremark{remark}{Remark}
\newsiamremark{hypothesis}{Hypothesis}
\newsiamremark{example}{Example}
\crefname{hypothesis}{Hypothesis}{Hypotheses}
\newsiamthm{claim}{Claim}

\def\G{{\mathcal G}}

\def\P{{\mathcal P}}

\def\cE{\mathbb{E}}

\def\ideal{{\mathsf{ideal}}}

\def\low{{\mathsf{low}}}

\def\sX{{\mathsf X}}

\def\cL{\mathcal{L}}
\def\R{\mathbb{R}}

\def\rZ{\mathbb{Z}}
\def\N{{\mathcal N}}

\def\Id{\mathbb{Id}}

\def\diag{\mathop{\rm diag}}

\def\Pr{\mathop{\rm Pr}}

\def\argmin{\mathop{\rm arg\, min}}

\allowdisplaybreaks

\headers{Best Mean Ergodic Averages via Optimal Graph Filters in Reversible MC}{N.Saldi}

\title{Best Mean Ergodic Averages via Optimal Graph Filters in Reversible Markov Chains\thanks{Submitted to the editors DATE.}}

\author{Naci Saldi\thanks{Department of Mathematics, Bilkent University, \c{C}ankaya, Ankara, Turkey (\email{naci.saldi@bilkent.edu.tr}).}}

\begin{document}

\maketitle

\begin{abstract}
In this paper, we address the problem of finding the best ergodic or Birkhoff averages in the mean ergodic theorem to ensure rapid convergence to a desired value, using graph filters. Our approach begins by representing a function on the state space as a graph signal, where the (directed) graph is formed by the transition probabilities of a reversible Markov chain. We introduce a concept of graph variation, enabling the definition of the graph Fourier transform for graph signals on this directed graph. Viewing the iteration in the mean ergodic theorem as a graph filter, we recognize its non-optimality and propose three optimization problems aimed at determining optimal graph filters. These optimization problems yield the Bernstein, Chebyshev, and Legendre filters. Numerical testing reveals that while the Bernstein filter performs slightly better than the traditional ergodic average, the Chebyshev and Legendre filters significantly outperform the ergodic average, demonstrating rapid convergence to the desired value. 
\end{abstract}

% REQUIRED

% REQUIRED

\section{Introduction}

The (mean) ergodic theorem is a cornerstone in the study of dynamical systems and stochastic processes, providing a foundation for understanding long-term average behavior. Traditional (mean) ergodic or Birkhoff averages are essential tools in this significant result, but their convergence rates can be slow. This paper seeks to enhance these averages using graph signal processing techniques, particularly through the application of graph filters.

Our approach begins by conceptualizing a function on the state space as a graph signal, where the (directed) graph, on which a graph signal can be defined, is constructed based on the transition probabilities of a reversible Markov chain. By introducing a notion of graph variation, we can define the graph Fourier transform for graph signals on this directed graph. This new perspective allows us to interpret the iteration in the mean ergodic theorem as a graph filter. However, this graph filter is not optimal in terms of convergence speed. To address this issue, we formulate three distinct optimization problems, each with a different objective, to derive the optimal graph filters. The solutions to these problems result in the Bernstein, Chebyshev, and Legendre polynomial filters.

We then validate our theoretical findings through numerical experiments. The results indicate that the Bernstein filter offers a slight improvement over the standard ergodic average. More notably, the Chebyshev and Legendre filters exhibit significantly better performance, achieving rapid convergence to the desired value. These findings highlight the effectiveness of our proposed graph signal processing approach in optimizing mean ergodic averages.

Our paper lies at the intersection of two major lines of research: acceleration techniques in numerical linear algebra based on approximation theory, and the study of optimal Birkhoff averages in the ergodic theory of deterministic dynamical systems. The first line seeks to speed up iterative algorithms for numerical linear algebra problems by designing polynomial filters that selectively attenuate undesirable spectral components \cite{Saa11}. The second investigates how to construct optimal ergodic averages that improve convergence rates in the ergodic theorem for deterministic systems \cite{DaYo18}.

In numerical linear algebra tasks such as solving linear systems or computing extreme eigenvalues are frequently accelerated using polynomial filtering techniques, either adaptive or non-adaptive. Classical examples include the Lanczos method \cite{Lan52,SaVo00}, which can be viewed as an adaptive polynomial filtering procedure for computing extreme eigenvalues of symmetric matrices, and its more general counterpart, the Arnoldi iteration \cite{Arn51}. In these methods, convergence is enhanced by adaptively reshaping the spectrum of the underlying matrix. Chebyshev acceleration, also known as the Chebyshev semi-iterative method \cite{GoVa61,Man77}, provides a complementary, non-adaptive approach, constructing polynomial filters that suppress unwanted spectral components while amplifying the desired ones \cite{Saa11}. Indeed, the optimization problem that gives rise to the Chebyshev filter in our setting coincides with the one underlying Chebyshev acceleration (or the Chebyshev semi-iterative method). This correspondence is entirely natural, as the construction of near-optimal polynomial filters through approximation-theoretic techniques is a classical problem that has been extensively developed in numerical analysis and signal processing. Accordingly, we do not claim novelty in this approximation-theoretic aspect, as in this part we simply apply well-established results from approximation theory within our problem framework. Our contribution instead lies in introducing a new perspective, namely the interpretation of the ergodic theorem through the framework of graph signal processing.

Specifically, we formulate the acceleration of ergodic averages as a graph signal processing problem. Functions acted upon by a Markov kernel are interpreted as graph signals on the directed graph induced by the transition probabilities. Within this framework, we introduce a notion of graph variation, relate it to the Laplacian spectrum of the graph, and interpret the Laplacian eigenvalues as frequencies. Under this viewpoint, ergodic convergence becomes a low-pass filtering problem: the zero Laplacian eigenvalue represents the lowest frequency, and the associated Fourier coefficient corresponds to the expectation with respect to the stationary distribution. The goal is therefore to eliminate higher-frequency components and retain only the zero-frequency mode.

This objective differs fundamentally from that of eigenvalue algorithms such as Lanczos and Arnoldi. Those methods are designed to compute extreme eigenvalues and their corresponding eigenvectors. In our setting, however, these quantities are already known: the smallest Laplacian eigenvalue is zero, and its eigenvector is the constant vector $\mathbf{1}$. Our aim is therefore not to compute this eigenpair, but rather to suppress all other spectral components of a given signal in order to extract its stationary mean.

There is also a structural distinction. Arnoldi-type methods are adaptive and require repeated orthogonalization, leading to substantial computational and memory costs. In contrast, our approach employs non-adaptive polynomial filters, similar in spirit to Chebyshev semi-iterative method. The filter is fixed in advance, depends only on spectral bounds, and can be implemented through simple iterative recursions that preserve the lightweight structure of the classical Birkhoff ergodic iteration. Indeed, as our goal is to generalize that iteration while maintaining its simplicity, adaptive procedures are not suitable for our purposes.

In summary, although polynomial spectral filtering is a classical technique and closely related to acceleration methods in numerical linear algebra, our contribution lies in the novel graph signal processing interpretation of ergodic convergence. By formulating ergodic acceleration as a low-pass filter design problem tailored to reversible Markov chains, we provide a fresh perspective that opens up new research directions. This viewpoint suggests that additional filtering techniques from the signal processing literature, such as IIR filters, could potentially be adapted to this setting to further enhance convergence. More broadly, framing the problem within signal processing allows us to draw upon a rich set of tools and methodologies from that field to address related challenges. We believe this perspective will inspire further work at the intersection of dynamical systems, Markov chains, and graph signal processing.

In the context of ergodic theory for deterministic dynamical systems, the Birkhoff ergodic theorem states that, for ergodic systems, the time averages of a function along a trajectory of length \( T \) converge to the corresponding space average as \( T \to \infty \). In practice, however, this convergence is often slow. To address this limitation, researchers have explored the use of non-uniform weights in place of the standard uniform averages that weight each point along the trajectory equally; an approach that is particularly relevant to our problem formulation.

In \cite{DaYo18}, the authors assign significantly lower weights to the initial and final points of the trajectory relative to those near the midpoint using a bump function. They demonstrate that, for quasi-periodic dynamical systems, these weighted averages achieve substantially faster convergence, provided the observable \( f \) is sufficiently differentiable. Related studies \cite{LiWe07, BeJol92, BeLo85} also investigate the acceleration of weighted ergodic sums, though they do not provide explicit convergence rates as in \cite{DaYo18}. In \cite{DuSc02}, a specific choice of weights yields a convergence rate that is inferior to that obtained in \cite{DaYo18}. More recently, \cite{ToYo22} establishes polynomial and exponential convergence rates for weighted Birkhoff averages of irrational rotations on tori, thereby improving upon the results in \cite{DaYo18}. Importantly, none of these works incorporate an optimization component into the problem formulation. They establish that if the weight-generating function satisfies certain properties, then the weighted averages converge to the space average at a certain speed; however, they do not claim optimality of the convergence rate nor that the chosen weights are optimal in any sense.

The most pertinent work to our study in this direction is \cite{RuBi24}. There, the authors expand the observable \( f \) in a basis and interpret weighted Birkhoff averages as a filter applied to the expansion coefficients. They observe that a filter tuned to specific frequencies can significantly outperform an arbitrary bump function of the type considered in \cite{DaYo18}. This perspective aligns closely with the approach we adopt in the present paper. Nevertheless, several important distinctions remain. First, \cite{RuBi24} considers deterministic dynamical systems on abstract topological spaces, whereas our focus is on stochastic dynamics governed by a Markov chain on a finite state space. Second, in \cite{RuBi24}, the expansion involves infinitely many frequency components, in contrast to the finite spectral decomposition that arises in our setting; consequently, the function approximated via filters differs, as infinitely many frequency components may contribute. Third, the basis functions corresponding to the frequencies in \cite{RuBi24} are not related to the eigenfunctions of the dynamical map, and thus no operator-theoretic or spectral analysis of the dynamics is employed. Finally, their work does not involve an underlying graph structure, whereas our formulation explicitly leverages the graph induced by the Markov transition kernel.

We note that even if one views optimal Birkhoff averages in deterministic dynamical systems via an operator theory perspective, a complication still arises as the operators are generally infinite-dimensional, and so their spectra are not well-behaved (isolated) as in the finite-dimensional case. Hence, graph signal processing techniques cannot be applied. Indeed, an important research problem in this direction is as follows: consider a Markov chain with an abstract state space. Then the transition probability can be viewed as an infinite-dimensional operator, where the spectrum may contain infinitely many elements. In this case, the eigenvalue $1$ is probably not isolated, and so we cannot directly apply graph filters as we cannot easily attenuate eigenvalues that are closer to $1$. In this case, we can first approximate the operator with a finite-dimensional one and then apply our graph filters.

\subsection{Our Contributions}

Standard ergodic averages are essential in mean ergodic theorem for Markov chains but they often suffer from slow convergence rates. This paper proposes a method to improve these averages by employing graph signal processing techniques, particularly through the use of graph filters.

\begin{itemize}

\item[(a)] In Section~\ref{graph_signal}, we initiate our method by interpreting a function on the state space as a graph signal. This graph signal is defined on a directed graph constructed from the transition probabilities of a reversible Markov chain. Introducing the concept of graph variation enables us to establish the graph Fourier transform for these signals on the directed graph. This framework allows us to view the iteration in the mean ergodic theorem as a graph filter. However, this graph filter does not achieve optimal convergence speed.

\item[(b)] To tackle the problem of suboptimal convergence speed, we introduce three distinct optimization problems in Sections~\ref{bernstein}, \ref{chebyshev}, and \ref{legendre}. Each optimization problem is designed with a unique objective in mind, aimed at deriving the most efficient graph filters for our purposes. The outcomes of these optimizations are the Bernstein, Chebyshev, and Legendre polynomial filters. By solving these optimization problems, we develop a set of graph filters that significantly enhance (at least for Chebyshev and Legendre polynomial filters) the convergence rates of mean ergodic averages.

\item[(c)] In Section~\ref{numerical}, we support our theoretical developments with numerical experiments. The results demonstrate that while the Bernstein filter provides a modest enhancement over the traditional ergodic average, the Chebyshev and Legendre filters show substantial performance improvements. 

\end{itemize}

\section{Reversible Markov Chains}
\label{rev-MC}

Let $\sX$ be a finite set and let $P:\sX\rightarrow\P(\sX)$ be a transition probability of a Markov chain on $\sX$. By an abuse of notation, we also denote the transition matrix via $P$, and so, $P \in \R^{\sX\times\sX}$. Therefore, $P(x,y) = \Pr\{X_{t+1} = y | X_t = x\}$ for all $x,y \in \sX$. We assume that $P$ is irreducible. In this case, it is known that there exists a unique stationary distribution $\pi$ of $P$ (see \cite[Corollary 1.17]{LePe17}) and $\pi(x) > 0$ for all $x \in \sX$ (see \cite[Proposition 1.19]{LePe17}). The transition probability $P$ is said to be \emph{reversible} if there exists a probability distribution $\pi$ on $\sX$ such that 
$
\pi(x) P(x,y) = \pi(y) P(y,x) \,\,\,\, \forall x,y \in \sX.
$
This equation is called \emph{detailed balance equation}. By \cite[Proposition 1.20]{LePe17}, any probability distribution that satisfies the detailed balance equation is a stationary distribution of $P$. Moreover, the corresponding Markov chain $\{X_t\}$ with initial distribution $\pi$ satisfies the following: for any $n$, we have
\begin{align*}
&\Pr\{X_0=x_0,X_1=x_1,\ldots,X_n=x_n\} 
= \pi(x_0) P(x_0,x_1)\ldots P(x_{n-1},x_n) \\
&= \pi(x_n) P(x_n,x_{n-1}) \ldots P(x_1,x_0) = \Pr\{X_0=x_n,X_1=x_{n-1},\ldots,X_n=x_0\}. 
\end{align*}
Hence, the time-reversed version of the Markov chain \(\{X_t\}\) is indistinguishable from the original chain, hence the term \emph{reversible} Markov chain. We define the following inner product on $\R^{\sX}$:
$$
\langle f,g \rangle_{\pi} \triangleq \sum_{x \in \sX} f(x) g(x) \pi(x). 
$$
The following result follows from \cite[Lemmas 12.1 and 12.2]{LePe17}. 

\begin{lemma}
Let $P$ be a reversible Markov chain with respect to $\pi$. Then, we have 
\begin{itemize}
\item[(1)] $P$ has real eigenvalues. We denote this set via $\sigma(P)$. 
\item[(2)] $\sigma(P) \subset [-1,1]$.
\item[(3)] There exists orthonormal basis of real-valued eigenfunctions $\{f_x\}_{x\in\sX}$ corresponding to real eigenvalues $\{\gamma_x\}_{x\in\sX}$ for inner product space $(\R^{\sX},\langle \cdot,\cdot \rangle_{\pi})$. Moreover, for $\gamma_{x^*} = 1$, we have $f_{x^*} = \mathbf{1}$.
\end{itemize}
\end{lemma}

The \emph{spectral gap} of a reversible Markov chain $P$ is defined as $\gamma \triangleq 1-\lambda_2$, where $\lambda_2$ is the second largest eigenvalue of $P$. We close this section by stating the Birkhoff ergodic theorem for irreducible Markov chains. 

\begin{theorem}
Let $P$ be an irreducible Markov chain with a unique stationary distribution $\pi$. Then, for any real-valued function $f$ on $\sX$, we have 
$$
\lim_{t\rightarrow\infty} \cE \left[\frac{1}{t} \sum_{k=0}^{t-1} f(X_t) \bigg| X_0 = x \right] = \sum_{z \in \sX} f(z) \pi(z) \,\, \forall x \in \sX,
$$ 
where $\{X_t\}$ is Markov chain with transition probability $P$. 
\end{theorem}

\noindent Note that we can write 
$$
\cE \left[\frac{1}{t} \sum_{k=0}^{t-1} f(X_t) \bigg| X_0 = x \right] = \frac{1}{t} (\Id+P+\ldots+P^{t-1})f(x). 
$$
Hence, ergodic theorem implies that  $\frac{1}{t} (\Id+P+\ldots+P^{t-1})f$ converges to the constant vector $\left(\sum_{z \in \sX} f(z) \pi(z)\right) {\bf 1}$ pointwise (or with respect to the norm induced by the inner product $\langle \cdot,\cdot \rangle_{\pi}$).

Let $\pi(f) = \sum_{z \in \sX} f(z) \pi(z)$. From a graph signal processing perspective, any function on $\sX$ is viewed as a graph signal residing in the space $(\mathbb{R}^{\sX}, \langle \cdot, \cdot \rangle_\pi)$, where the directed graph structure is induced by the transition probability $P$. The Graph Fourier Transform (GFT) is then defined as
\[
\hat f(x) = \langle f, f_x \rangle_\pi \in (\mathbb{R}^{\sX}, \langle \cdot, \cdot \rangle),
\]
where $\{f_x\}_{x \in \sX}$ are the eigenvectors of the Laplacian, with eigenvalues $\{\lambda_x\}_{x \in \sX} \subset [0,2]$. 

\begin{figure}[H]
\begin{center}
\scalebox{0.75}{
\begin{tikzpicture}[
    node distance=3cm,
    box/.style={draw, minimum width=2.8cm, minimum height=1.2cm},
    arrow/.style={-Latex, thick}
]

\node (f) {$f(x)$};
\node[box, right=of f] (gft) {GFT};
\node[right=of gft] (fh) {$\hat f(x)=\langle f,f_x\rangle_\pi$};

\draw[arrow] (f) -- (gft);
\draw[arrow] (gft) -- (fh);

\end{tikzpicture}
}
\end{center}
\caption{Graph Fourier Transform}
\end{figure}

A large eigenvalue $\lambda_x$ indicates strong signal variation (high frequency), whereas a small $\lambda_x$ (low frequency) corresponds to low variation (see Lemma~\ref{variation}). In the extreme case where $\lambda_{x^*} = 0$, we have $f_{x^*} = \mathbf{1}$, representing no variation at all. Since our goal is to compute \(\pi(f) = \langle f, f_{x^*} \rangle_{\pi}\), we should attenuate all high-frequency components except for \(\lambda_{x^*} = 0\); in other words, we need to design a low-pass filter.

\begin{figure}[H]
\begin{center}
\scalebox{0.75}{
\begin{tikzpicture}[
    node distance=3.5cm,
    box/.style={draw, minimum width=4cm, minimum height=3cm},
    arrow/.style={-Latex, thick}
]

% Input
\node (fh) {$\hat f(x)$};

% Filter box
\node[box, right=of fh] (lp) {};

% Output
\node[right=of lp] (out) {Filtered signal $\pi(f)$};

\draw[arrow] (fh) -- (lp);
\draw[arrow] (lp) -- (out);

% Title inside box
\node at ([yshift=1.15cm]lp.center) {\small Low-Pass Filter};

% Draw frequency response inside box
\begin{scope}
\clip (lp.south west) rectangle (lp.north east);

% Axes
\draw[thick] 
([xshift=-1.6cm,yshift=-0.9cm]lp.center) -- 
([xshift=1.6cm,yshift=-0.9cm]lp.center);

\draw[thick] 
([xshift=-1.6cm,yshift=-0.9cm]lp.center) -- 
([xshift=-1.6cm,yshift=0.9cm]lp.center);

% Low-pass curve
\draw[thick]
([xshift=-1.6cm,yshift=0.7cm]lp.center) --
([xshift=-0.4cm,yshift=0.7cm]lp.center) --
([xshift=0.6cm,yshift=0.1cm]lp.center) --
([xshift=1.6cm,yshift=-0.6cm]lp.center);

% Axis labels
\node at ([xshift=1.5cm,yshift=-1.1cm]lp.center) {\tiny $\lambda$};
\node at ([xshift=-1.7cm,yshift=1.1cm]lp.center) {\tiny $H(\lambda)$};

\end{scope}
\end{tikzpicture}
}
\end{center}
\caption{Low Pass Filter}
\end{figure}

The iteration in the ergodic theorem,
$$
\frac{(\Id+P+\ldots+P^{t-1})}{t} f
$$
can then be interpreted as a low-pass polynomial graph filter applied to the graph signal $f$. However, this particular polynomial filter is not optimal with respect to certain well-known objective functions. Motivated by this limitation, and following the classical signal processing literature \cite{KaMc95}, this paper develops algorithms to identify the optimal low-pass polynomial graph filter; one that achieves faster convergence to $\pi(f) \mathbf{1}$ than the standard ergodic average. The mathematical formulation of this approach is presented in the next section.

\section{Graph Signal Processing Approach to Birkhoff Ergodic Theorem}\label{graph_signal}

In this section, $\G = (\sX, E)$ denotes a directed graph with a vertex set $\sX$ (i.e. state space) and an edge set $E$. Therefore, if $(x, y) \in E$, there is a directed edge from $x$ to $y$, denoted as $x \rightarrow y$. Each directed edge in this graph has an associated weight given by our reversible (with respect to the distribution $\pi$) transition probability $P(x, y) \in \R_+$. Hence, $P$ satisfies the detailed balance equation:
$
\pi(x) P(x,y) = \pi(y) P(y,x) \,\, \forall x,y \in \sX.
$ 
In view of this construction, we endow $\R^{\sX}$ with the inner product $\langle \cdot, \cdot \rangle_{\pi}$ introduced in the previous section. We direct the reader to the book \cite{Ort22} for the terminology and notation related to graph signal processing used in this paper.

In graph $\G$, for any vertex $x$, its out-degree is $d_{out}(x) = \sum_{x \rightarrow y} P(x, y) = 1$. Consequently, the out-degree matrix for this graph is the identity matrix $\Id$. Therefore, the combinatorial graph Laplacian is defined as $L \triangleq \Id - P$. Given that $\sigma(P) \subset [-1,1]$, it follows $\sigma(L) \subset [0,2]$.

%\begin{remark}
%Given that $P$ satisfies the detailed balance equation with respect to $\pi$, we can define the following symmetric matrix:
%$$
%A(x,y) \triangleq \pi(x)^{1/2} P(x,y) \pi(y)^{-1/2} \,\, \forall x,y \in \sX. 
%$$
%In matrix notation, if we let $D \triangleq \diag(\pi)$ (i.e., the diagonal matrix with diagonal entries formed by the vector $\pi$), then, $A = D^{1/2} P D^{-1/2}$. Hence, $A$ and $P$ are similar matrices. Similar matrices represent the same linear operator in different bases, and they share all properties of the underlying operator, including the identical set of (real) eigenvalues.
%Consequently, we can use $A$ as a weight matrix instead of $P$. In this case, since $A$ is symmetric, we can treat our graph as a weighted undirected graph. However, a complication arises because the stationary distribution $\pi$ is either unknown or difficult to compute precisely, making it challenging to construct $A$ for use in graph signal processing.  
%\end{remark}

Now, it is time to define signals on graphs in the node domain. Later, we will describe graph signals in the frequency domain.

\begin{definition}
A graph signal $f$ is a mapping from  the vertex set $\sX$ to the reals.
\end{definition} 

In graph signal processing, the frequency of a signal is directly related to its variation across neighboring vertices. High signal variation indicates high frequency, while low variation suggests low frequency. To discuss these concepts accurately, we need to define the variation of a graph signal in our context. The next definition will provide this.

\begin{definition}
A variation of a graph signal $f$ is defined as
$$
TV_{f} \triangleq \frac{\left(\sum_{x \in \sX} \sum_{x \rightarrow y} \pi(x) P(x,y) |f(x) - f(y)|^2\right)^{1/2}}{\left(\sum_{x \in \sX} \pi(x) |f(x)|^2\right)^{1/2}}.
$$
\end{definition} 

In this definition, the denominator serves as a normalization constant to obtain a dimensionless quantity. The numerator quantifies the weighted difference of signals on vertices from their neighbors, with weights provided by the matrix $DP$, where $D \triangleq \diag(\pi)$. The value $\pi(x) P(x, y)$ represents how similar the vertices $x$ and $y$ are. Here, $\pi(x)$ indicates the importance of vertex $x$ in terms of the stationary distribution. If $\pi(x) \approx 0$, vertex $x$ holds little importance, as the number of visits to this vertex in the Markov chain is very low compared to other vertices according to the ergodic theorem. Additionally, $P(x, y)$ quantifies the (directed) similarity between vertex $x$ and $y$ or it can be interpreted as the importance of $y$ to $x$ because it is proportional to the number of transitions from vertex $x$ to $y$. Thus, the term $\pi(x) P(x, y) |f(x) - f(y)|^2$ captures the importance of vertex $x$ multiplied by the weighted variation of the graph signal $f$ from $x$ to $y$, with the significance of this variation determined by directed similarity $P(x, y)$. 

The following result is important for interpreting the eigenvalues of $L$ as frequencies. 

\begin{lemma}\label{variation}
We have $\sqrt{2 \langle f, Lf \rangle_{\pi}} = TV_{f}$ for any graph signal $f$ with unit norm $\|f\|_{\pi} = 1$. 
\end{lemma} 

\begin{proof}
The proof of this lemma is very similar to \cite[Remark 3.1]{Ort22}. For any $x \in \sX$, we can write $Lf(x) = f(x) - \sum_{x \rightarrow y} P(x,y) f(y) = \sum_{x \rightarrow y} P(x,y) (f(x)-f(y))$ as $\sum_{x \rightarrow y} P(x,y)=1$. Hence, 
\small
\begin{align*}
2 \langle f, Lf \rangle_{\pi} 
&= \sum_{x \in \sX} \sum_{x \rightarrow y} \pi(x) P(x,y) f(x) (f(x)-f(y)) + \sum_{y \in \sX} \sum_{y \rightarrow x} \pi(y) P(y,x) f(y) (f(y)-f(x)) \\
&= \sum_{x \in \sX} \bigg(\sum_{x \rightarrow y} \pi(x) P(x,y) f(x) (f(x)-f(y)) + \sum_{y \rightarrow x} \pi(y) P(y,x) f(y) (f(y)-f(x)) \bigg)
\end{align*}
\normalsize
as $x \rightarrow y$ if and only if $y \rightarrow x$ by reversibility. Since $\pi(x) P(x,y) = \pi(y) P(y,x)$, we have 
$
\pi(y) P(y,x) f(y) = \pi(x) P(x,y) f(y). 
$
Hence, 
\begin{align*}
\sum_{x \rightarrow y} \pi(x) P(x,y) f(x) (f(x)-f(y))  
&+ \sum_{y \rightarrow x} \pi(y) P(y,x) f(y) (f(y)-f(x)) \\
&= \sum_{x \rightarrow y} \pi(x) P(x,y) (f(x)-f(y))^2.
\end{align*}
This implies that 
$
2 \langle f, Lf \rangle_{\pi} = \sum_{x \in \sX} \sum_{x \rightarrow y} \pi(x) P(x,y) (f(x)-f(y))^2 = TV_f^2
$
as $\|f\|_{\pi} = 1$. 
\end{proof}

This result implies that $\langle f, Lf \rangle_{\pi}$ is high when the variation of $f$ is high, and low when the variation of $f$ is low. Therefore, we can use $\langle f, Lf \rangle_{\pi}$ to quantify the variation of graph signal $f$. Let us write the eigenvalues of $L$ as $\{\lambda_x\}_{x \in \sX} \subset [0,2]$, where one of them is zero, say $\lambda_{x_*} = 0$. These eigenvalues have corresponding real eigenfunctions $\{f_x\}_{x \in \sX}$, which form an orthonormal basis for the space $(\R^{\sX},\langle\cdot,\cdot\rangle_{\pi})$ and $f_{x_*} = {\bf 1}$. For any eigenfunction $f_x$, the following holds:
$$
\langle f_x, Lf_x \rangle_{\pi} = \langle f_x, \lambda_x f_x \rangle_{\pi} = \lambda_x. 
$$
Therefore, if $\lambda_x$ is high, the variation of the corresponding eigenfunction $f_x$ is also high. This implies that eigenfunctions associated with high eigenvalues exhibit high frequencies. A similar conclusion can be drawn for eigenfunctions associated with low eigenvalues. The lowest frequency is $\lambda_{x_*} = 0$, which corresponds to the constant eigenfunction ${\bf 1}$, where the variation is zero. Hence, we can view eigenvalues $\{\lambda_x\}_{x \in \sX}$ of $L$ as frequencies. 

\subsection{Graph Filters}

For any vertex $x \in \sX$, its $k$-hop neighborhood, $\N_k(x)$, is the set all vertices that are part of a directed path with no more than $k$ edges starting from $x$. Now, it is time to give the definition of a graph filter. 

\begin{definition}
A (linear) graph filter $H$ is a linear operator on graph signals $f$ on $\sX$. 
\end{definition}

An operator $H$ on graph signals is called $k$-hop operator if $Hf(x)$ can be computed via the elements of $\N_k(x)$ for all $x \in \sX$. Hence,  $P$ and $L$ are one-hop operators. Let $Z$ represent a generic one-hop operator, such as $Z = P$, $Z = L$, or $Z = A$, among others. In graph signal processing, $Z$ serves as a generalization of the shift operator used in classical signal processing. For any positive integer $k$, $Z^k$ is a $k$-hop operator. Initially, our definition of a graph filter required only that the operator $H$ be linear, allowing any arbitrary linear operator to be chosen for $H$. However, we now focus on filters that can be expressed as polynomials of $Z$.

\begin{definition}
Given a one-hop operator $Z$, the polynomial graph filter $H$ is a linear operator of the following form $H = p(Z)$ for some polynomial $p$ of some degree $K$. Hence,
$
H = \sum_{k=0}^K a_k Z^k,
$
where $Z^0  = \Id$.  
\end{definition} 

To design a polynomial filter, one should first select a specific one-hop operator $Z$. In our case, this one-hop operator will be the combinatorial Laplacian $L$ since we have a precise interpretation of its eigenvalues as frequencies. Recall that the eigenvalues of $L$ are $\{\lambda_x\}_{x \in \sX} \subset [0,2]$, where $\lambda_{x_*} = 0$, with the corresponding real eigenfunctions $\{f_x\}_{x \in \sX}$ forming an orthonormal basis for the space $(\R^{\sX}, \langle \cdot, \cdot \rangle{\pi})$ and $f_{x_*} = \mathbf{1}$. Thus, any graph signal $f$ has the following unique representation:
$
f = \sum_{x \in \sX} \langle f,f_x \rangle_{\pi} f_x. 
$
This representation leads to the definition of graph Fourier transform.
\begin{definition}
The vector ${\hat f} \triangleq (\langle f,f_x \rangle_{\pi})_{x \in \sX}$ in $\R^{\sX}$, where $\R^{\sX}$ is now endowed with the usual Euclidean inner product $\langle \cdot,\cdot \rangle$, is defined as the graph Fourier transform (GFT) of the graph signal $f$. Let us denote the GFT operator via $\cL$ that performs this operation, and so, $\cL(f) = {\hat f}$. 
\end{definition}

Obviously $\cL$ is a linear operator from $(\R^{\sX},\langle\cdot,\cdot\rangle_{\pi})$ to $(\R^{\sX},\langle\cdot,\cdot\rangle)$. Moreover, for any $x \in \sX$, we have $\cL(f_x) = e_x$, where $e_x(y) \triangleq 1_{\{x = y\}}$. Hence, $\cL$ maps orthonormal basis $\{f_x\}_{x \in \sX}$ of $(\R^{\sX},\langle\cdot,\cdot\rangle_{\pi})$ to orhonormal basis $\{e_x\}_{x \in \sX}$ of $(\R^{\sX},\langle\cdot,\cdot\rangle)$. Therefore, it is a unitary transformation. Hence, the inverse of the GFT $\cL$ is its adjoint: $\cL^{-1} = \cL^*$, where $\cL^*$ is the unique operator from  $(\R^{\sX},\langle\cdot,\cdot\rangle)$ to $(\R^{\sX},\langle\cdot,\cdot\rangle_{\pi})$ satisfying the following:
$$
\langle \cL(f),{\hat g} \rangle = \langle f,\cL^*({\hat g}) \rangle_{\pi} \,\, \forall f \in (\R^{\sX},\langle\cdot,\cdot\rangle_{\pi}), {\hat g} \in (\R^{\sX},\langle\cdot,\cdot\rangle). 
$$
Hence $\cL^{-1}({\hat g})(x) = \sum_{y \in \sX} f_y(x) {\hat g}(y)$. 

Let $H$ be a polynomial graph filter of degree $K$; that is, $H = \sum_{k=0}^K a_k L^k \triangleq p(L)$. This filter then acts on any graph signal $f$ with GFT ${\hat f}$ in the following manner:
$$
H(f) = \sum_{x \in \sX} {\hat f}(x) \, p(\lambda_x) \, f_x. 
$$
Hence, in the frequency domain, the polynomial filter $H$ can be described as follows: 
$$
{\hat H} \triangleq \diag(\{p(\lambda_x)\}_{x \in \sX}) = p(\Lambda),
$$ 
where $\Lambda \triangleq \diag(\{\lambda_x\}_{x \in \sX})$; that is, $\Lambda$ is the diagonal matrix containing the eigenvalues of $L$ as its diagonal entries. Given ${\hat H}$, we can describe the action of the polynomial filter $H$ in the frequency domain as follows:
$
{\hat H}({\hat f}) = \diag(\{{\hat f}(x) \, p(\lambda_x)\}_{x \in \sX}) = \cL(H(f)). 
$
Unfortunately, since $\pi$ is generally unavailable, it is in general unfeasible to compute the GFT of any graph signal $f$ and frequency representation of any polynomial graph filter $H$. Consequently, filtering operations must be performed in the node domain instead of the frequency domain. Nevertheless, we can still gain insight into designing graph filters in the node domain by examining their effects in the frequency domain.

Now, we provide an informal definition of the problem we aim to solve in this paper. A precise mathematical description will follow later.

\begin{tcolorbox} 
[colback=white!100]
%\begin{center}
{\bf Informal Definition of the Problem:} For each degree $K$, find the optimal polynomial filter $H_K$ with the corresponding polynomial $p_K$ such that $H_K(f)$ converges to $\pi(f){\bf 1}$ as quickly as possible. To achieve this convergence, the polynomial $p_K$ should attenuate the eigenvalues (or frequencies) other than $0$ of $L$. Therefore, the polynomial graph filter should function as a low-pass filter. 
%\end{center}
\end{tcolorbox}

\subsection{Revisiting Birkhoff Ergodic Theorem and Problem Formulation}\label{birkhoffsection}

\hspace{1pt} Note that the iteration in the ergodic theorem can be written as a polynomial of $P$ in the following form:
$
\frac{(\Id+P+\ldots+P^{t-1})}{t} \triangleq p_t(P),
$
where 
$p_t(z) \triangleq \sum_{k=0}^{t-1} \frac{1}{t} z^k = \frac{1}{t} \frac{1-z^t}{1-z}, \,\,\, z \in [-1,1).$
If we define the following polynomial $q_t(z) \triangleq p_t(1-z)$, then we can write the same expression in terms of $L$ as follows:
\begin{align*}
p_t(P) &= q_t(L) = \frac{1}{t} \frac{1-(1-L)^t}{1-(1-L)} = \frac{1}{t} \frac{1-\sum_{k=0}^t {t \choose k} (-L)^k}{L} \\
&= \sum_{k=1}^t \frac{1}{t} {t \choose k} (-1)^{k-1} L^{k-1} \triangleq \sum_{k=0}^{t-1} a_e(k) \, L^{k}. 
\end{align*}
If we define the polynomial graph filter $H_t \triangleq q_t(L)$, then we observe that we are indeed applying the polynomial graph filter $H_t$ to the graph signal $f$ in each iteration of the ergodic theorem. As $t$ approaches infinity, this graph filter attenuates the frequencies different from zero. Specifically,
$$
q_t(z) \rightarrow 0 \,\, \text{as} \,\, t \rightarrow \infty
$$
for $z \in (0,2]$ and $q_t(0)=1$. This can be proved very easily since $p_t(z) \rightarrow 0$ as $t \rightarrow \infty$ for $z \in [-1,1)$. Consequently, $H_t(f) \rightarrow \pi(f) {\bf 1}$. Hence, in each iteration of the ergodic theorem, we are applying a low-pass filter to the graph signal $f$. However, this low-pass filter is not the optimal one. The goal of this paper is to find the optimal low-pass filter in each iteration. This problem is very similar to the problem of designing optimal low-pass filters in classical signal processing \cite{KaMc95}. Therefore, as in classical signal processing, we use results from approximation theory.

Given the non-zero eigenvalues $\{\lambda\}_{x \in \sX \setminus \{x_*\}} \subset (0,2]$ of $L$, it is not the case that $q_t(\lambda_x) = 0$ for all $x \in \sX \setminus \{x_*\}$ for sufficiently large $t$ values, where $q_t$ is the polynomial used in the ergodic theorem. However, if we have complete knowledge of the eigenvalues $\{\lambda\}_{x \in \sX}$ of $L$, then we can use a Lagrange polynomial $q$ of degree $|\sX|-1$ \cite[Definition 6.1]{VeKoGo14} to achieve:
$$
q(\lambda_{x_*})=1 \,\, \text{and} \,\, q(\lambda_x) = 0 \,\, \text{for} \,\, x \neq x_*.
$$
In this case, we set $q_t^L(z) = (1 + z + \ldots + z^{t - |\sX|}) \, q(z)$ for $t \geq |\sX|$. Hence 
$$
q_t^L(L) \, f = \sum_{x \in \sX} {\hat f}(x) \, q_t^L(\lambda_x) \, f_x = \pi(f) {\bf 1}.
$$ 
That is, we can filter out frequencies that are different from zero. However, in general, we do not have access to the full set of eigenvalues of $L$ and so constructing $q_t^L$ is unfeasible. 

In general, the most that one can do is to estimate some important eigenvalues different than $\lambda_{x_*}$. For instance, we may obtain some lower bound to the second smallest eigenvalue of $L$ \cite{DiSt91} \cite[Section 9.2]{Bre01}. Indeed, if $\lambda_{x_{**}}$ denotes the second smallest eigenvalue of $L$, then $\lambda_{x_{**}}$ is equal to the spectral gap of $P$. In this case, if we have some lower bound on this second smallest eigenvalue:
$
\lambda_{x_{**}} \geq \lambda_{\low},
$
then the best approach would be to design an ideal low-pass filter:
$$
q^{\ideal}(\lambda_{x_*}) = 1 \,\, \text{and} \,\, q^{\ideal}(\lambda) = 0 \,\, \text{for} \,\, 2 \geq \lambda \geq \lambda_{\low}. 
$$
However, achieving this is impractical because an ideal low-pass filter demands an infinite number of components, as highlighted in classical signal processing \cite{VeKoGo14}. Therefore, rather than using an infinite-degree polynomial, we seek a polynomial of degree $t$ at each time step $t$ that meets the following criteria:
$$
q_t(\lambda_{x_*}) = 1 \,\, \text{and} \,\, q_t(\lambda) \approx 0 \,\, \text{for} \,\, 2 \geq \lambda \geq \lambda_{\low}. 
$$
Here, we also want to converge to zero for values in $[\lambda_{\low},2]$ as quickly as possible while maintaining $q_t(\lambda_{x_*}) = 1$. In essence, we aim to closely approximate the ideal low-pass filter in the most efficient manner. To achieve this, we can view our problem as designing an optimal polynomial low-pass graph filter with a one-hop operator $L$ to eliminate frequencies above $\lambda_{\low}$, where $\lambda_{\low}$ is a fixed threshold. This problem shares strong parallels with the design of optimal FIR low-pass filters in classical signal processing \cite{KaMc95}. Indeed, in our work, we rely on concepts and results from approximation theory that are quite similar to those used in optimal FIR low-pass filter design. 

Let $P_K[\lambda_{\low},2]$ be the set of all polynomials of degree at most $K$. Then, at each time step $t\geq2$, we want to solve the following optimization problem:
\begin{tcolorbox} 
[colback=white!100]
%\begin{center}
$$
{\bf (OPT_t)} \min_{p \in P_{t-1}[\lambda_{\low},2]} \|p\|_{\infty}^{[\lambda_{\low},2]} \,\, \text{subject to} \,\, p(0) = 1.
$$
\end{tcolorbox} 
\noindent Here
$$\|p\|_{\infty}^{[\lambda_{\low},2]} \triangleq \sup_{z \in [\lambda_{\low},2]} |p(z)|.$$ 

\begin{remark}
The optimization problem described above is precisely the same as the one underlying Chebyshev acceleration, also known as the Chebyshev semi-iterative method \cite{GoVa61, Man77, Saa11}. Mathematically, therefore, our optimal graph filter problem is equivalent to the filtering problem encountered in the Chebyshev semi-iterative method. This equivalence is not surprising, as such optimal filtering problems are well established in the numerical linear algebra and signal processing literature. Indeed, the solution is given by a classical result from approximation theory concerning Chebyshev polynomials. Consequently, from a purely mathematical standpoint, this problem has already been extensively treated across various communities, beginning with approximation theory and extending to signal processing. We therefore make no claim of novelty in this regard; we are simply applying well-known results from approximation theory. The same holds for the other filter designs considered in this paper (Bernstein and Legendre).

Our contribution lies instead in the novel perspective we bring to this problem through the lens of graph signal processing. This reinterpretation is significant because it opens the door to applying a broader range of well-established filtering techniques, such as IIR filters, to further accelerate convergence. Moreover, by framing the problem within signal processing, one can now draw upon a rich set of tools and methodologies from that field to address related challenges in this setting. 
\end{remark}

Let us now give a motivation for this optimization problem. Note that, originally, we want to minimize $\|p_t(L)f-\pi(f){\bf 1}\|_{\pi}$ for any unit norm graph signal $f$, where $p_t \in P_{t-1}[\lambda_{\low},2]$ with $p_t(0)=1$, for each time step $t\geq2$ as any graph signal is viewed as an element of the inner product space $(\R^{\sX},\langle\cdot,\cdot\rangle_{\pi})$\footnote{In this context, we can use any norm we prefer (for instance, sup-norm) because the interpretation at the end will be the same.}. An upper bound to this term can be obtained as follows:
\begin{align*}
&\|p_t(L)f-\pi(f){\bf 1}\|_{\pi} = \left\| \sum_{x \in \sX} {\hat f}(x) p_t(\lambda_x) f_x - {\hat f}(x_*) f_{x_*} \right\|_{\pi} 
= \left\| \sum_{x \neq x_*} {\hat f}(x) p_t(\lambda_x) f_x \right\|_{\pi} \\
&\leq \sum_{x \neq x_*} |{\hat f}(x)| |p_t(\lambda_x)| \left\| f_x \right\|_{\pi} 
= \sum_{x \neq x_*} |{\hat f}(x)| |p_t(\lambda_x)| 
\leq \left(\sum_{x \neq x_*} |{\hat f}(x)|\right) \sup_{x \neq x_*} |p_t(\lambda_x)| 
\end{align*}  
In the last inequality above, the first term depends on the graph signal $f$, and so, we cannot control this. Thus, our objective should be to minimize the second term by choosing $p_t$ appropriately. However, because we do not have access to the full eigenvalue structure, we cannot directly minimize the second term. Given the lower bound $\lambda_{\low}$ for the second smallest eigenvalue of $L$, we can instead try to minimize $\|p_t\|_{\infty}^{[\lambda_{\low}, 2]}$, which leads to the optimization problem ${\bf (OPT_t)}$.

\section{Filter Design with Bernstein Polynomials}\label{bernstein}

\hspace{1pt} For each time step $t\geq2$, we want to design a polynomial $p$ of degree at most $t-1$ so that 
$$
p(\lambda_{x_*}) = 1 \,\, \text{and} \,\, p(z) \approx 0 \,\, \text{for} \,\, 2 \geq z \geq \lambda_{\low}. 
$$
In this sub-section, instead of pursuing the optimal design formulated in ${\bf (OPT_t)}$, we approach the problem differently. Firstly, the behavior of the  ideal low-pass filter over the interval $(0,\lambda_{\low})$ is not critical. Therefore, let us fix a continuous function $g$ on $[0,2]$ such that $g(\lambda_{x_*}) = g(0) = 1$ and $g(z) = 0$ for $2 \geq z \geq \lambda_{\low}$; that is, $g$ can be considered as an ideal low-pass filter given $\lambda_{\low}$. We have the freedom to choose the behavior of $g$ on the interval $(0,\lambda_{\low})$ as desired. This choice is also part of the design process. Note that Weierstrass approximation theorem \cite{Riv81} states that any continuous function over compact domain can be approximated via polynomials under sup-norm. This theorem is generally proved constructively using Bernstein polynomials (see the proof of \cite[Theorem 1.1]{Riv81}). Therefore, we can use Bernstein polynomials to approximate our ideal low-pass filter $g$ over $[0,2]$. To this end, let us first define Bernstein polynomials. 

\begin{definition}
Given any function $h$ that is bounded on $[0,1]$, we define its Bernstein polynomial of degree $K$ by 
$
B_K(h;z) \triangleq \sum_{l=0}^K h\left(\frac{l}{K}\right) {K \choose l} z^l \, (1-z)^{K-l}.
$
\end{definition}

One can prove that (see \cite[Theorem 1.2]{Riv81})
$
\|h-B_K(h)\|_{\infty}^{[0,1]} \leq \frac{3}{2} \, \omega_h\left(\frac{1}{\sqrt{K}}\right), 
$
where $\omega_h$ is the modulus of continuity of $h$. Since our function $g$ is defined over $[0,2]$ and Bernstein polynomials are defined over $[0,1]$, we need to translate this discussion into the interval $[0,2]$. The Bernstein polynomial of degree $K$ of $g$ is defined as 
$$
B_K(g;z) \triangleq \sum_{l=0}^K g\left(\frac{2l}{K}\right) {K \choose l} \left(\frac{z}{2}\right)^l \, \left(1-\frac{z}{2}\right)^{K-l}.
$$
Then, by above result, we have
$$
\|g-B_K(g)\|_{\infty}^{[0,2]} \leq \frac{3}{2} \, \omega_g\left(\frac{2}{\sqrt{K}}\right). 
$$
Additionally, $B_K(g;\lambda_{x_*}) = B_K(g;0) = 1$. Thus, $B_K(g)$ is a strong candidate for our polynomial filter. However, given the above rate of convergence result, to achieve a sharp decrease in the function value of $B_K(g)$ over the interval $[\lambda_{\low},2]$, it makes sense to select $g$ or adjust its behavior on the interval $(0,\lambda_{\low})$ so that $\omega_g$ is minimized at $\frac{2}{\sqrt{K}}$. Consequently, we propose the following optimization problem for each time step $t \geq 2$:
\begin{tcolorbox} 
[colback=white!100]
%\begin{center}
\begin{align*}
{\bf (OPT^B_t)} &\min_{g \in C[0,2]} \omega_g\left(\frac{2}{\sqrt{t-1}}\right) \,\, \text{subject to} \,\, g(0) = 1, \,\, g([\lambda_{\low},2]) = 0.
\end{align*}
\end{tcolorbox}  

Solving this optimization problem for any $t \geq 2$ is generally challenging and impractical to perform at each time step $t \geq 2$. However, as $t$ approaches infinity, we can demonstrate that the triangle function $g_{\Delta}$, which decreases from 1 to zero as we move from 0 to $\lambda_{\low}$, is proved to be asymptotically optimal for this problem. Therefore, for each $t \geq 2$, we can select our ideal low-pass filter $g$ as $g_{\Delta}$ to design our Bernstein filter.

\begin{tcolorbox} 
[colback=white!100]
\begin{align*}
&\text{{\bf Asymptotically Optimal Bernstein Filter}} \\ &B_{K}(g_{\Delta};z) \triangleq \sum_{l=0}^K g_{\Delta}\left(\frac{2l}{K}\right) {K \choose l} \left(\frac{z}{2}\right)^l \, \left(1-\frac{z}{2}\right)^{K-l}
\end{align*}
\end{tcolorbox}

Let us now demonstrate that the triangle function is asymptotically optimal. First, the modulus of continuity, denoted as $\omega_{\Delta}$, for the triangle function $g_{\Delta}$ is given by the following:
$$
\omega_{\Delta}(\delta) = \begin{cases} m \delta & \text{if} \,\, \delta \in [0,\lambda_{\low}) \\ 1 & \text{otherwise} \end{cases}
$$
where $m \triangleq 1/\lambda_{\low}$. Note that if $2/\sqrt{t-1} \geq \lambda_{\low}$, then for any $g \in C[0,2]$ satisfying the constraints in ${\bf (OPT^B_t)}$, we have $\omega_g\left(2/\sqrt{t-1}\right) = 1$. Hence, any feasible $g \in C[0,2]$ solves ${\bf (OPT^B_t)}$. For this reason, in the following, we assume that $2/\sqrt{t-1} < \lambda_{\low}$.

\begin{proposition}
If $g \in C[0,2]$ is an optimal solution of ${\bf (OPT^B_t)}$, where $ r \triangleq 2/\sqrt{t-1} < \lambda_{\low}$, then 
$
\omega_{\Delta}(r) \frac{n+\lambda}{n+1+\lambda} \leq \omega_g(r) < \omega_{\Delta}(r).
$
Here, $\lambda_{\low} = n r + \lambda r$ and $\lambda \in (0,1)$. Hence, as $t \rightarrow \infty$, we have $\lambda \rightarrow 0$, $n \rightarrow \infty$, and  $\frac{n+\lambda}{n+1+\lambda} \rightarrow 1$. This implies that $g_{\Delta}$ is asymptotically optimal.
\end{proposition}

\begin{proof}
Let $g \in C[0,2]$ be an optimal solution\footnote{If the optimal solution does not exist, we can use an $\varepsilon$-optimal solution for a sufficiently small $\varepsilon > 0$. In this case, we should replace $\omega_{\Delta}$ with $\omega_{\Delta} + \varepsilon$ in the second inequality stated in the proposition.} of ${\bf (OPT^B_t)}$, where $r < \lambda_{\low}$. By \cite[Lemma 1.3]{Riv81}, we have 
\begin{align*}
1 = \omega_{g}(nr+\lambda r) \leq (n+\lambda+1) \omega_g(r)  
\leq (n+\lambda+1) \omega_{\Delta}(r).
\end{align*}
But note that $1 = \omega_{\Delta}(nr+\lambda r) = (n+\lambda) \omega_{\Delta}(r)$. Hence
$
(n+\lambda)  \omega_{\Delta}(r) \leq (n+\lambda+1) \omega_g(r)  \leq (n+\lambda+1) \omega_{\Delta}(r).
$
Dividing both sides with $(n+\lambda+1)$ establishes the result.
\end{proof}

Another justification for why the triangle function is a good candidate for ideal low-pass filter is the following. Instead of considering the space $C[0,2]$, we can work with the set of continuously differentiable, except possibly at $\lambda_{\low}$, functions on $[0,2]$, denoted as $C^1[0,2]$. Since estimating the modulus of continuity exactly can be challenging, we can use the inequality $\omega_g(\delta) \leq \|g^{(1)}\|_{\infty}^{[0,2]} \delta$ for all $\delta \geq 0$ to obtain an upper bound on the modulus of continuity. This allows us to state a more manageable optimization problem, which does not depend on time.
\begin{tcolorbox} 
[colback=white!100]
%\begin{center}
\begin{align*}
{\bf (\widehat{OPT}_B)} &\min_{g \in C^1[0,2]} \|g^{(1)}\|_{\infty}^{[0,2]} \,\,\text{subject to} \,\, g(0) = 1, \,\, g([\lambda_{\low},2]) = 0.
\end{align*}
\end{tcolorbox}

The solution of the last problem is again the triangle function $g_{\Delta}$ as stated in the next result.

\begin{proposition}
The triangle function $g_{\Delta}$ is an optimal solution of ${\bf (\widehat{OPT}_B)}$.
\end{proposition}

\begin{proof}
According to the mean value theorem, for any function $g \in C^1[0,2]$ satisfying $g(0) = 1$ and $g([\lambda_{\low},2]) = 0$, there exists an $l \in (0,\lambda_{\low})$ such that $|g^{(1)}(l)| = 1/\lambda_{\low}$, and so, $\|g^{(1)}\|_{\infty}^{[0,2]} \geq 1/\lambda_{\low}$. Furthermore, note that for the triangle function, $\|g^{(1)}_{\Delta}\|_{\infty}^{[0,2]} = 1/\lambda_{\low}$. 
\end{proof}

\section{Filter Design with Chebyshev Polynomials}\label{chebyshev}

\hspace{1pt} Recall that for each time step $t\geq2$, we want to solve the following optimization problem:
\begin{tcolorbox} 
[colback=white!100]
%\begin{center}
$$
{\bf (OPT_t)} \min_{p \in P_{t-1}[\lambda_{\low},2]} \|p\|_{\infty}^{[\lambda_{\low},2]} \,\, \text{subject to} \,\, p(0) = 1.
$$
\end{tcolorbox} 
In this section, we establish that the optimal solution of 
${\bf (OPT_t)} $ is a properly normalized Chebyshev polynomial of degree $t-1$, for any $t \geq 2$. Since we have explicit forms of Chebyshev polynomials of any degree, we can easily implement the corresponding filter. Before proving this result, let us briefly discuss Chebyshev polynomials and outline one of their key properties, which will lead us to the optimal solution of ${\bf (OPT_t)}$. This property can be stated as follows: among all polynomials with a fixed degree and a sup-norm less than $1$ over any interval $[a,b]$, Chebyshev polynomials are the ones that increase as quickly as possible compared to other polynomials outside this interval. Using this property, it becomes clear that Chebyshev polynomials are the optimal solutions of ${\bf (OPT_t)}$.  

Chebyshev polynomials are originally defined over the interval $[-1,1]$. Therefore, we begin by defining these polynomials over this interval and stating their maximum growth property. Then, we translate everything to the interval $[\lambda_{\low},2]$ in our problem formulation. 

\begin{definition}
The Chebyshev polynomial of degree $K$ is defined as follows:
$
T_K(z) \triangleq \cos (K \theta),
$
where $z = \cos(\theta)$ and $\theta \in [0,\pi]$. 
\end{definition}

Note that since $\cos (K \theta)$ is a polynomial of $\cos(\theta)$ with degree $K$ \cite[Exercise 1.6, p. 43]{Riv81}, $T_K$ is well-defined. Moreover, since it is a polynomial, once we know its coefficients, we know it as a function for all $z$. The following result states that Chebyshev polynomials grows as rapidly as possible outside of $[-1,1]$. 

\begin{theorem}{\cite[Theorem 1.10, p. 31]{Riv81}}\label{growth}
If $p \in P_K[-1,1]$ and $|z_0| \geq 1$, then 
$
|p^{(l)}(z_0)| \leq \|p\|_{\infty}^{[-1,1]} \, |T_K^{(l)}(z_0)|, \,\, l=0,1,\ldots,K,
$
and in particular, we have 
$
|p(z_0)| \leq \|p\|_{\infty}^{[-1,1]} \, |T_K(z_0)|.
$
\end{theorem}

Using above maximal growth condition, we can now prove the main result of this section. Before doing that, let us define the Chebyshev polynomial of degree $K$ over the interval $[\lambda_{\low},2]$ using a linear transformation as follows: 
$$
T_K^{[\lambda_{\low},2]}(z) \triangleq T_K\left(\frac{2z-2-\lambda_{\low}}{2-\lambda_{\low}}\right).
$$
Hence, $T_K^{[\lambda_{\low},2]}(\lambda_{\low}) = T_K(-1)$ and $T_K^{[\lambda_{\low},2]}(2) = T_K(1)$. Note that $T_K^{[\lambda_{\low},2]}(0) \neq 1$, and so, $T_K^{[\lambda_{\low},2]}$ is not a feasible point for ${\bf (OPT_t)}$ in this form. Hence, we need to normalize it as follows:
$$
p_K^*(z) \triangleq \frac{T_K^{[\lambda_{\low},2]}(z)}{T_K^{[\lambda_{\low},2]}(0)}.
$$
Now, we have $p_K^*(0) = 1$. Moreover, we also have 
$$
\|p_{K}^*\|_{\infty}^{[\lambda_{\low},2]} = \frac{\left\|T_K^{[\lambda_{\low},2]}\right\|_{\infty}^{[\lambda_{\low},2]}}{\left|T_K^{[\lambda_{\low},2]}(0)\right|} = \frac{1}{\left|T_K^{[\lambda_{\low},2]}(0)\right|}
$$

\begin{theorem}
We have $p_{t-1}^* \in \argmin {\bf (OPT_t)}$ for any $t \geq 2$. Hence, the optimal polynomial filter is given by properly normalized Chebyshev polynomial of degree $t-1$ at each time step $t \geq 2$. 
\end{theorem}

\begin{proof}
Fix any $t \geq 2$. Suppose that there exists $p \in P_{t-1}[\lambda_{\low},2]$ such that $p(0) =1$ and $\|p\|_{\infty}^{[\lambda_{\low},2]} < \|p_{t-1}^*\|_{\infty}^{[\lambda_{\low},2]}$. Our goal is to reach a contradiction. Note that Theorem~\ref{growth} is still true for Chebyshev polynomials over the interval $[\lambda_{\low},2]$. Therefore, as $0 \notin [\lambda_{\low},2]$, we have 
$
|p(0)| \leq \|p\|_{\infty}^{[\lambda_{\low},2]} \left|T_K^{[\lambda_{\low},2]}(0)\right|. 
$ 
But this implies that 
$
|p(0)| < \|p_{t-1}^*\|_{\infty}^{[\lambda_{\low},2]} \left|T_K^{[\lambda_{\low},2]}(0)\right| = 1
$
as 
$
\|p_{t-1}^*\|_{\infty}^{[\lambda_{\low},2]} = \frac{1}{\left|T_K^{[\lambda_{\low},2]}(0)\right|}. 
$
This is a contradiction as we assumed $p(0)=1$.
\end{proof}

Note that  Chebyshev polynomials over $[-1,1]$ satisfy the following recursive relation \cite[p. 55]{Riv81}:
$
T_{K+1}(z) = 2z \, T_{K}(z) - T_{K-1}(z) \,\, K\geq1. 
$
Hence we have 
$$
T_{K+1}^{[\lambda_{\low},2]}(z) = \frac{2z-2-\lambda_{\low}}{2-\lambda_{\low}} \, T_K^{[\lambda_{\low},2]}(z) - T_{K-1}^{[\lambda_{\low},2]}(z).
$$
For each $K\geq1$, let us define 
$$
\alpha_K \triangleq \frac{T_K^{[\lambda_{\low},2]}(0)}{T_{K+1}^{[\lambda_{\low},2]}(0)}, \,\, \beta_K \triangleq \frac{T_{K-1}^{[\lambda_{\low},2]}(0)}{T_{K+1}^{[\lambda_{\low},2]}(0)}.
$$
Then, we have the following recursive relation between optimal Chebyshev filters:
\begin{tcolorbox} 
[colback=white!100]
\begin{align*}
&\text{{\bf Recursive Implementation of Chebyshev Filter}} \\ &p_{K+1}^*(L) = \alpha_K \, \frac{2L-2-\lambda_{\low}}{2-\lambda_{\low}}  p_K^*(L) - \beta_K \, p_{K-1}^*(L)
\end{align*}
\end{tcolorbox} 
This formula enables us to calculate the output of the optimal filter recursively, significantly reducing the amount of algebraic computation required in each iteration, as it eliminates the need to compute \( p_{K+1}^*(L) f \) from scratch. By storing \( p_K^*(L)f \) and \( p_{K-1}^*(L)f \) in memory, we can easily calculate \( p_{K+1}^*(L)f \) using the recursive formula.

\section{Filter Design with Legendre Polynomials}\label{legendre}

\hspace{1pt} In this section, we reformulate the optimization problem ${\bf (OPT_t)}$ by using the $l_2$-norm instead of the sup-norm. It turns out that the optimal solution is weighted combination of properly scaled and normalized Legendre polynomials. 

For each time step $t \geq 2$, our goal is to solve the following optimization problem:
\begin{tcolorbox} 
[colback=white!100]
%\begin{center}
$$
{\bf (\widehat {OPT}_t)} \min_{p \in P_{t-1}[\lambda_{\low},2]} \|p\|_{2}^{[\lambda_{\low},2]} \,\, \text{subject to} \,\, p(0) = 1.
$$
\end{tcolorbox} 
Here, the $l_2$-norm is with respect to the Lebesgue measure over.
In general, for any Lebesgue integrable weight function $w:[\lambda_{\low},2] \rightarrow (0,\infty)$, we can define the weighted $l_2$-norm as 
$$
\|p\|_{2,w}^{[\lambda_{\low},2]} \triangleq \left( \int_{[\lambda_{\low},2]} |p(x)|^2 \, w(x) \, m(dx) \right)^{1/2}.
$$
Therefore, depending on $w$, the solution of the optimization problem $\mathbf{(\widehat{OPT}_t)}$ varies. However, because we lack knowledge of the eigenvalues of $P$ in the interval $[\lambda_{\low}, 2]$, we do not know how to assign weights to the points within this interval. Thus, it is more reasonable to assign equal weights, i.e., $w(x) = 1$, to all the points in $[\lambda_{\low}, 2]$, which leads to the classical $l_2$-norm.

To solve the above optimization problem, we need to define orthonormal  polynomials over $[\lambda_{\low}, 2]$ with respect to the classical $l_2$-norm. 
Note that if our interval is $[-1,1]$, then Legendre polynomials $\{L_n\}$ are the orthogonal polynomials \cite[p. 53]{Riv81}, where these polynomials are recursively defined as follows:
\begin{align*}
L_0(z) = 1, \,\, L_1(z) = z, \,\, 
(n+1) L_{n+1}(z) = (2n+1)zL_n(z)-nL_{n-1}(z) \,\, n\geq1. 
\end{align*}
Note that for each $n\geq 0$, we have \cite[Exercise 2.11, p. 62]{Riv81}
$$
\int_{[-1,1]} L_n(z)^2 m(dz) = \frac{2}{2n+1}.
$$
Hence, to obtain an orthonormal system, we need to normalize the Legendre polynomials as follows:
$$
\hat L_n(z) \triangleq \sqrt{n+\frac{1}{2}} \, L_n.
$$
These normalized Legendre polynomials satisfy the following recursive formula:
$$
\sqrt{\frac{2(n+1)^2}{2n+3}} \hat L_{n+1}(z) = \sqrt{2(2n+1)} z \hat L_n(z) - \sqrt{\frac{2n^2}{2n-1}} \hat L_{n-1}(z).
$$
Now, we can transform the orthonormality properties of these polynomials to the interval $[\lambda_{\low}, 2]$ by simply applying a linear transformation, which leads to the following orthonormal (with respect to the classical $l_2$-norm) polynomials over $[\lambda_{\low}, 2]$: 
$$
\tilde L_n(z) \triangleq \frac{2}{2-\lambda_{\low}} \, \hat L_n\left(\frac{2z-2-\lambda_{\low}}{2-\lambda_{\low}}\right) \,\, n \geq 0.
$$
Therefore, they satisfy the following recursive formula:
\begin{align*}
\sqrt{\frac{2(n+1)^2}{2n+3}} \tilde L_{n+1}(z) 
= \sqrt{2(2n+1)} \, \frac{2z-2-\lambda_{\low}}{2-\lambda_{\low}} \, \tilde L_n(z) - \sqrt{\frac{2n^2}{2n-1}} \tilde L_{n-1}(z).
\end{align*}
If we define 
$$
\alpha_n \triangleq \sqrt{\frac{2n^2}{2n-1}}, \,\, \beta_n \triangleq \sqrt{2(2n+1)} \,\, n\geq 1,
$$
then we can write the above recursive formula as
$$
\alpha_{n+1} \tilde L_{n+1}(z) = \beta_n \,  \frac{2z-2-\lambda_{\low}}{2-\lambda_{\low}} \, \tilde L_n(z) - \alpha_n \tilde L_{n-1}(z).
$$

Now, we present the main result of this section. To do so, for each $t\geq2$ and $k\geq0$, we define
$$
\xi_k^{t-1} \triangleq \frac{\tilde L_k(0)}{\sum_{k=0}^{t-1} \tilde L_k(0)^2}.
$$ 
Using these coefficients, for each $t\geq2$, let us also define the following polynomial:
$$
p_{t-1}^*(z) \triangleq \sum_{k=0}^{t-1} \xi_k^{t-1} \, \tilde L_k(z).
$$

\begin{theorem}
We have $p_{t-1}^* \in \argmin {\bf (\widehat {OPT}_t)}$ for any $t \geq 2$. Hence, the optimal polynomial filter is given by weighted combination of properly normalized and scaled Legendre polynomials of degree at most $t-1$, at each time step $t \geq 2$. 
\end{theorem}

\begin{proof}
The proof of this problem is quite classical and can be found in any book on orthogonal polynomials. We refer readers to \cite[p. 78]{Tot05}. For completeness, we provide the full proof here, which is quite elementary.

Note that any polynomial $p \in P_{t-1}[\lambda_{\low},2]$ can be written 
$
p = \sum_{k=0}^{t-1} c_k \, \tilde L_k.
$
If $p$ is a feasible point of ${\bf (\widehat {OPT}_t)}$, then 
$
p(0) = \sum_{k=0}^{t-1} c_k \, \tilde L_k(0) = 1.
$
Hence, by Cauchy's inequality, we have 
$$
1 \leq \left(\sum_{k=0}^{t-1}c_k^2\right) \left(\sum_{k=0}^{t-1}\tilde L_k(0)^2\right).
$$
Therefore, for any feasible $p$, we have 
$$
\left(\|p\|_{2}^{[\lambda_{\low},2]}\right)^2 = \sum_{k=0}^{t-1} c_k^2  \geq \frac{1}{\left(\sum_{k=0}^{t-1}\tilde L_k(0)^2\right)}.
$$
Note that we have equality above if and only if 
$$
c_k = \frac{\tilde L_k(0)}{\sum_{k=0}^{t-1} \tilde L_k(0)^2} \triangleq \xi_k^{t-1}.
$$
This completes the proof.
\end{proof}

Note that we can implement this optimal filter recursively at two different time scales. Specifically, for each $K\geq0$, we define
$$
\gamma_K \triangleq \frac{\sum_{k=0}^{K} \tilde L_k(0)^2}{\sum_{k=0}^{K+1} \tilde L_k(0)^2}.
$$
Then the optimal filter can be computed via the following coupled iterations:

\begin{tcolorbox} 
[colback=white!100]
\begin{align*}
&\hspace{-50pt}\text{{\bf Recursive Implementation of \bf Legendre Filter}} \\
p_{K+1}^*(L) &= \gamma_K \, p_K^*(L) + \xi_{K+1}^{K+1} \tilde L_{K+1}(L) \\
\tilde L_{K+1}(L) &=  \frac{\beta_K}{\alpha_{K+1}}  \frac{2L-2-\lambda_{\low}}{2-\lambda_{\low}} \tilde L_K(L) - \frac{\alpha_K}{\alpha_{K+1}} \tilde L_{K-1}(L)
\end{align*}
\end{tcolorbox} 

This formula allows us to compute the output of the optimal filter recursively, reducing the algebraic computation needed in each iteration, as it avoids recomputing \( p_{K+1}^*(L) f \) from the beginning. By storing \( p_K^*(L)f \), \( \tilde L_{K}(L)f \), and \( \tilde L_{K-1}(L)f \) in memory, we can efficiently compute \( p_{K+1}^*(L)f \) using the coupled recursive formulas.

\section{Numerical Experiments}\label{numerical}

In this section, we present numerical experiments to illustrate the performance of the optimal filters derived in the previous sections. Our primary aim is to compare these filters with the standard iterative procedure from the classical ergodic theorem. To this end, we consider two simple yet representative examples: the random walk on graphs and the Glauber chain. These examples are chosen precisely because they capture the essential features of our theoretical framework while remaining transparent and easily interpretable. As our goal is to demonstrate the effectiveness of the method in a clear and accessible manner, we deliberately keep the setup simple. Extending these techniques to larger-scale and real-world applications is a natural next step, and we leave such implementations to practitioners who may build on the foundations laid in this work.

\subsection{Random Walk on a Graph}

Let $\G = (\sX,E)$ be an undirected, simple, and connected graph with the vertex set $\sX$ and the edge set $E$. A random walk on $\G$ is described as follows: if $X_t = x$ is the current position, then the next position is chosen uniformly from the set of neighbors of $x$. Hence, if $d(x)$ is the degree of $x$, then 
$$
P(x,y) = \begin{cases} \frac{1}{d(x)} & \text{if} \,\, \{x,y\} \in E \\ 0 & \,\, \text{otherwise} \end{cases}
$$ 
Since the graph is connected, the corresponding Markov chain is irreducible. It is also straightforward to see that the Markov chain is reversible with respect to 
$
\pi(x) = \frac{d(x)}{2|E|}, \,\, x \in \sX,
$
where $|E|$ is the number of edges. 

In the numerical experiments, we consider a very simple random walk on a cycle \cite[Example 2.1]{DiSt91}. Let $p$ be an odd number. We place the integers $\rZ_p \triangleq \{0,1,2,\ldots,p-1\}$ on a circle in the given order and connect each consecutive point in the circle. This forms our graph. In Figure~\ref{fig:circle_integers}, an example is given where $p=11$.  

\begin{figure}[H]
    \centering
\begin{tikzpicture}
% Parameters
\def\p{11} % Number of integers (must be odd)
\def\radius{1.2cm} % Radius of the circle

% Draw the circle
\draw (0,0) circle (\radius);

% Place the integers around the circle
\foreach \i in {0,...,\numexpr\p-1} {
    \node[draw, circle, inner sep=1pt, fill=white] (N\i) at ({360/\p * \i}:\radius) {\i};
}
\end{tikzpicture}
\caption{A circle with \( p \) integers placed evenly around it, where \( p \) is an odd number.}
\label{fig:circle_integers}
\end{figure}

We consider a random walk on this graph. Therefore, we have 
$$
P(x,y) = \begin{cases} \frac{1}{2} & \text{if} \,\, |x-y| \in \{1,p-1\} \mod p \\ 0 & \,\, \text{otherwise} \end{cases}
$$
and $\pi(x) = 1/p$. For this Markov chain, by \cite[Corollary 1]{DiSt91} (see also \cite[p. 44]{DiSt91}), one can obtain the following lower bound to the second smallest eigenvalue $\lambda_{x_{**}}$ of $L$ as follows:
$$
\lambda_{x_{**}} \geq \lambda_{\low} \triangleq \frac{8p}{(p-1)^2(p+1)}.
$$
Indeed, for this random walk, it is possible to compute the exact eigenvalues of \( L \) for small $p$ values. However, since this is a toy example in which we want to demonstrate the effectiveness of our method, and computing the eigenvalues of \( L \) for other, more complicated models is not feasible, we use the above estimate instead of the exact value.

In the numerical experiment, we take $p=11$ and generate a random function $f$ as follows:
\small
\[
f = \begin{bmatrix}
8.53, 6.22, 3.50, 5.13, 4.01, 0.75, 2.39, 1.23, 1.83, 2.39, 4.17 
\end{bmatrix}
\]
\normalsize
In this case, $\lambda_{\low} = 0.0733$. We apply the ergodic average, Bernstein polynomial filter, Chebyshev polynomial filter, and Legendre polynomial filter up to degree $20$, and obtain the results depicted in Figure~\ref{random_walk_on_graphs}.

\begin{figure}[h]
    \centering
    \includegraphics[width=7cm]{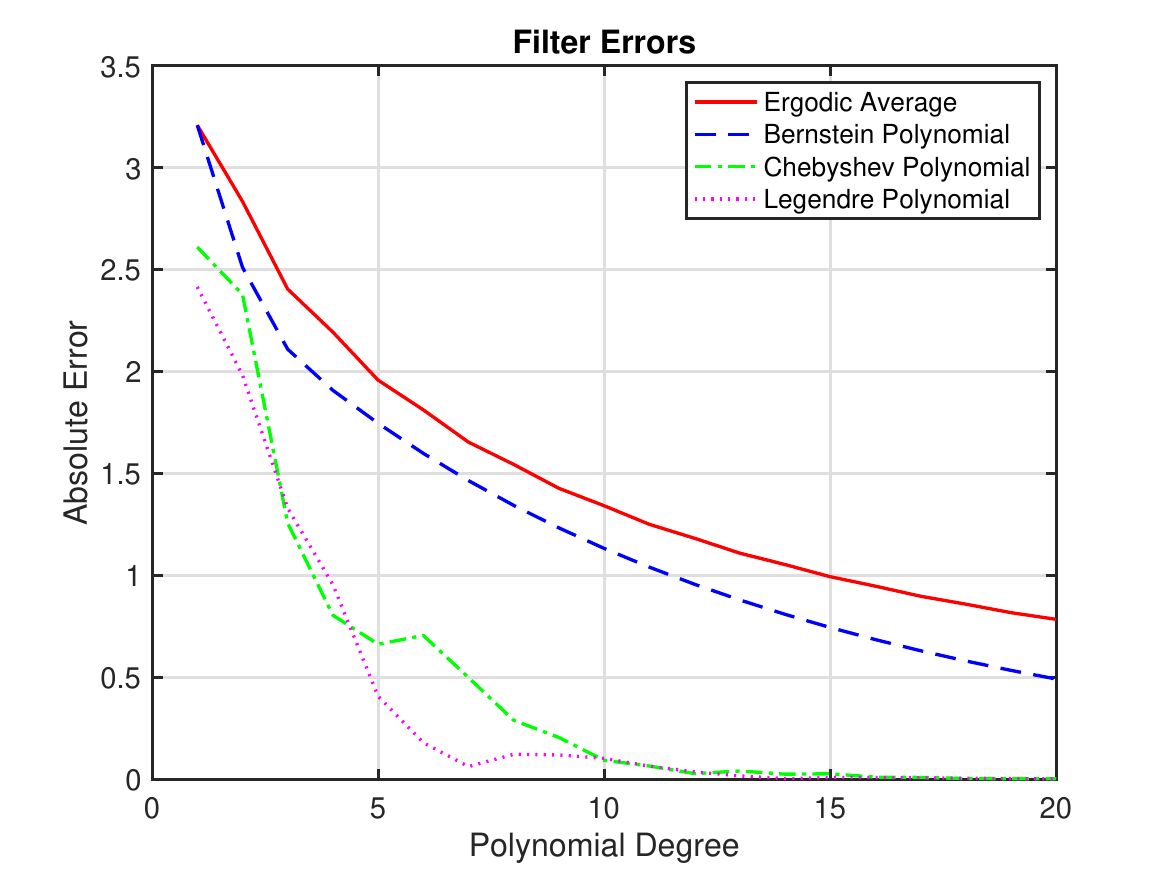}
    \caption{A maximum absolute error as a function of degree of the filter.}
    \label{random_walk_on_graphs}
\end{figure}

In this figure, the maximum absolute error is computed as 
$$
\max_{x \in \sX} \bigg| p(L)f(x) - \sum_{x \in \sX} f(x) \, \pi(x) \bigg|
$$
for any polynomial filter $p(L)$. This gives the worst deviation of the filter output from the desired value with respect to the initial position. As one can see, the Bernstein polynomial filter is slightly better than the ergodic average. However, the Chebyshev and Legendre polynomial filters outperform the ergodic average significantly, as expected. Indeed, for these polynomials, the maximum absolute error decreases to zero very quickly. This demonstrates the usefulness of the optimality of these filters with respect to the sup-norm and \( l_2 \)-norm compared to other polynomial filters.

\subsection{Glauber Chain}

Let $V$ and $S$ be finite sets. We can think of $V$ as the set of vertices of a graph and $S$ is the set of possible values that can be taken by these vertices. Our state space $\sX$ is a subset of $V^S$. Let $\pi$ be a probability distribution on $\sX$. In statistical physics, $\pi$ is called a Gibbs distribution. 

The \emph{Glauber dynamics} for $\pi$ is a reversible Markov chain with state space $\sX$ and stationary distribution $\pi$. The transition probability is given as follows: if $X_t = x$ is the current state, then we generate a point (vertex) $v$ from $V$ uniformly and pick the next state randomly via
$$
X_{t+1} = y \sim \pi(y|\sX(x,v)) = \begin{cases} \frac{\pi(y)}{\pi(\sX(x,v))}& \text{if} \,\, y \in \sX(x,v) \\ 0& \text{otherwise}. \end{cases}
$$
Here, $\sX(x,v) \triangleq \{y \in \sX: y(w) = x(w) \,\, \forall \, w \neq v\}$. 
Verbally, a new state is chosen according to the probability distribution $\pi$ conditioned on the set of states equal to the current state $x$ at all vertices different from $v$. Define $\sX(x) \triangleq \bigcup_{v \in V} \sX(x,v)$. Then, the transition probability can be defined as follows:
\begin{align*}
P(x,y) = \begin{cases}\frac{1}{|V|} \pi(y|\sX(x,v))& \text{if} \,\, y \in \sX(x) \, \text{and} \, y \neq x \\ \sum_{v \in V} \frac{1}{|V|}\pi(y|\sX(x,v))& \text{if} \,\, y=x \\ 0& \text{otherwise}.\end{cases}
\end{align*} 
Note that if $\pi(x)>0$ for all $x \in \sX$, then $P$ is irreducible. Moreover, one can also prove that $P$ and $\pi$ satisfy the detailed balance condition. Hence, $\pi$ is a stationary distribution of $P$ and $P$ is reversible. 

 In the numerical experiment, we consider a Glauber chain on a cycle \cite{Ser03}. Let $p$ be the number of vertices in the cycle. Each vertex in this graph can take two values: $+1$ or $-1$. Hence, the state space of this Markov chain is $\sX \triangleq \{+1,-1\}^p$. In Figure~\ref{fig:glauber_circle}, an instance of a state transition of this Markov chain is given where $p=11$. In these two cycles, only the direction of the top nodes are different, and so, with positive probability, a transition can occur.

\begin{figure}[H]
    \centering
\begin{subfigure}
\centering
\begin{tikzpicture}
    \def \n {11} % Number of vertices (must be odd)
    \def \radius {1.2cm} % Radius of the circle
% Draw the circle
\draw (0,0) circle (\radius);
% Define positions and spins
    \foreach \s in {1,...,\n}
    {
        \pgfmathparse{\s>3&&\s<8} % Randomly assign 0 or 1
        \let\spin\pgfmathresult
        \node[draw, circle, inner sep=0.2pt, fill=white] (A\s) at ({360/\n*(\s - 1)}:\radius) {\ifnum\spin=0 $\downarrow$ \else $\uparrow$ \fi};
    }
    
\draw[->] (1.5,0) to[out=45, in=135] (3,0);
\end{tikzpicture}
\end{subfigure} 
\begin{subfigure}
\centering
\begin{tikzpicture}
    \def \n {11} % Number of vertices (must be odd)
    \def \radius {1.2cm} % Radius of the circle
% Draw the circle
\draw (0,0) circle (\radius);
% Define positions and spins
    \foreach \s in {1,...,\n}
    {
        \pgfmathparse{\s>4&&\s<8} % Randomly assign 0 or 1
        \let\spin\pgfmathresult
        \node[draw, circle, inner sep=0.2pt, fill=white] (A\s) at ({360/\n*(\s - 1)}:\radius) {\ifnum\spin=0 $\downarrow$ \else $\uparrow$ \fi};
    }
\end{tikzpicture}
\end{subfigure} 
\caption{A state transition of a Glauber chain on the cycle. Up arrow means $+1$, down arrow means $-1$. }
\label{fig:glauber_circle}
\end{figure}

In this model, the Gibbs distribution is given by
$$
\pi(x) = \frac{e^{-\beta H(x)}}{Z(\beta)},
$$
where $H(x) = - \sum_{u,w \in \rZ_p : u \sim w} J_{v,w} \, x(v) x(w)$ is the energy of the state $x$ and $Z(\beta)$ is the partition function. Therefore, we have 
$$
P(x,y) = \frac{1}{p} \sum_{w \in \rZ_p} \frac{e^{\beta y(w) S(x,w)} \, 1_{\{x(v) = y(v) \,\, \forall v \neq w\}}}{e^{\beta y(w) S(x,w)}+e^{-\beta y(w) S(x,w)}} ,
$$
where $S(x,w) \triangleq \sum_{u: u \sim w} J_{u,w} \, x(u)$. For this Markov chain, we can obtain its second largest eigenvalue exactly via computing eigenvalues of some low-dimensional matrix \cite[Corollary 6]{Ser03}. Indeed, let $M$ be $p \times p$ matrix whose entries are all zero except immediately above and below the diagonal and non-zero entries have the following description: 
\begin{align*}
M(i,i-1) = s_{i-1}/(c_{i-1}+c_i), \,\,
M(i,i+1) = s_{i}/(c_{i-1}+c_i),
\end{align*} 
where $J_i \triangleq J_{i,i+1}$, $s_i \triangleq \sinh(2J_i)$, and $c_i \triangleq \cosh(2J_i)$. Let the eigenvalues of $M$ be $\gamma_1 \geq \gamma_2 \geq \ldots \geq \gamma_p$. Then, by \cite[Corollary 6]{Ser03}, the second largest eigenvalue of $P$ is given by 
$
1-\frac{1-\gamma_1}{p}.
$
Hence, we can take $\lambda_{\low}$ as 
$
\lambda_{x_{**}} = \frac{1-\gamma_1}{p} \triangleq \lambda_{\low}.
$

In the numerical experiment, we take $p=4$, $J_{i} = 1$ for all $i$, and $\beta = 0.2$. We generate a random function $f$ as follows:
\small
\begin{align*}
f = [
9.04,  9.79, 4.38, 1.11, 2.58, 4.08, 5.94, 2.62, 6.02, 7.11, 
2.21, 1.17, 2.96, 3.18, 4.24, 5.07]
\end{align*}
\normalsize
In this case, $\lambda_{\low} = 0.155$. We apply the ergodic average, Bernstein polynomial filter, Chebyshev polynomial filter, and Legendre polynomial filter up to degree $20$, and obtain the results depicted in Figure~\ref{glauber_chain}.

\begin{figure}[h]
    \centering
    \includegraphics[width=7cm]{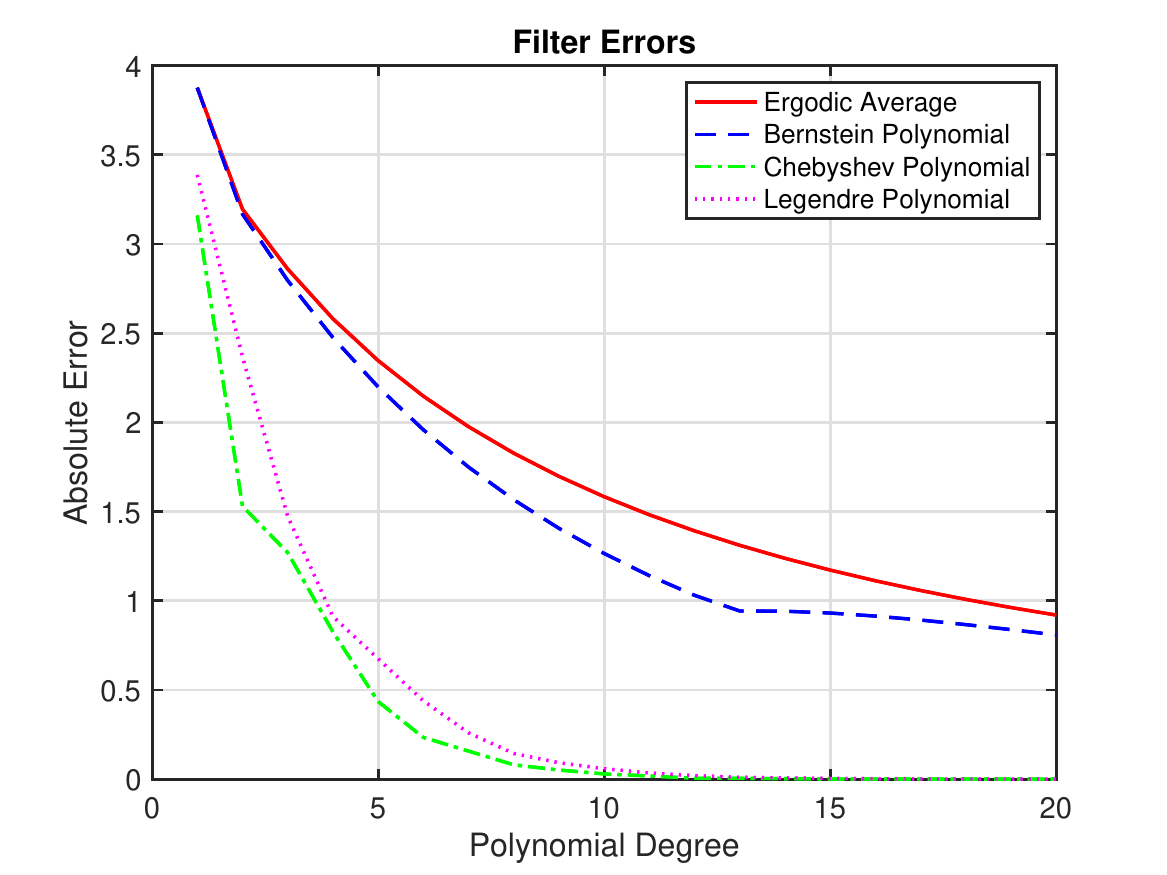}
    \caption{A maximum absolute error as a function of degree of the filter.}
    \label{glauber_chain}
\end{figure}

In this figure, the maximum absolute error is computed as in the previous example. It is evident that the Bernstein polynomial filter exhibits a slight improvement over the ergodic average. However, both the Chebyshev and Legendre polynomial filters markedly surpass the ergodic average, as anticipated. In fact, for these polynomials, the maximum absolute error diminishes rapidly towards zero similar to the previous example.

\section{Conclusion} \label{conclusion}

In this paper, we explore the optimization of ergodic or Birkhoff averages in the context of the ergodic theorem for reversible Markov chains via utilizing graph signal processing techniques to achieve faster convergence to desired values. We use the transition probabilities of a Markov chain to form a directed graph so that we can represent a function on the state space as a graph signal. This allowed us to introduce a graph variation concept and define the graph Fourier transform for graph signals on this directed graph. We identify the inherent non-optimality of the standard graph filter used in ergodic theorem iterations and address this by formulating three optimization problems with distinct objectives. The solutions to these problems yield the Bernstein, Chebyshev, and Legendre polynomial filters, each offering improved performance over traditional method. Our numerical experiments demonstrate that the Bernstein filter provides a modest improvement over the ergodic average. More significantly, the Chebyshev and Legendre filters show substantial enhancements, achieving rapid convergence to the desired value. 

\subsection*{Future Research Directions}

We have identified several promising research directions to extend the work presented in this paper. A primary area of interest involves the analysis of non-reversible Markov chains. In this setting, the graph Laplacian is no longer symmetric, causing its spectrum to move into the complex plane. This shift presents a significant mathematical challenge: the design of optimal graph filters would necessitate the use of complex polynomials. Consequently, the desirable properties and optimality guarantees associated with Bernstein, Chebyshev, and Legendre polynomials, which are central to our current framework, may no longer hold. Pursuing this direction requires a deeper exploration into the approximation theory of complex polynomials.

A further research direction involves extending this problem to abstract state spaces. In this setting, the classical graph structure no longer exists; however, analogous to deterministic dynamical systems, spectral filtering methods may still be applied to accelerate ergodic averages. To adapt the optimal filtering techniques from our graph signal processing framework to this continuous or infinite context, a key challenge arises: the spectral domain, which may contain countably many elements (or even a continuous spectrum), must be approximated by a finite set. This approximation must satisfy a specific spectral gap property to ensure that the filters remain effective. Bridging this gap between discrete graph signals and abstract operator theory would allow for a more generalized application of optimal filtering.

A final research direction involves the investigation of Infinite Impulse Response (IIR) filters, which correspond to rational filtering methods. In this framework, the acceleration is achieved using rational functions rather than polynomials.
Implementing this approach requires a deep dive into the approximation properties of rational functions, analogous to the well-studied roles of Bernstein, Chebyshev, and Legendre polynomials in the FIR (Finite Impulse Response) setting. Successfully extending the theory to rational filters would allow for a rigorous performance comparison, enabling us to identify the specific trade-offs, advantages, and limitations of rational versus polynomial filtering for ergodic acceleration.
\section{Acknowledgements}

The author thanks Professor Alexander Goncharov for guiding him to important findings in approximation theory.

\bibliographystyle{siamplain}
\bibliography{references}

\begin{thebibliography}{10}

\bibitem{Arn51}
{\sc W.~E. ARNOLDI}, {\em The principle of minimized iterations in the solution
  of the matrix eigenvalue problem}, Quarterly of Applied Mathematics, 9
  (1951), pp.~17--29.

\bibitem{BeLo85}
{\sc A.~Bellow and V.~Losert}, {\em The weighted pointwise ergodic theorem and
  the individual ergodic theorem along subsequences}, Transactions of the
  American Mathematical Society, 288 (1985), pp.~307--345.

\bibitem{BeJol92}
{\sc J.~R.~L. Bellow, Alexandra}, {\em Almost everywhere convergence of
  weighted averages.}, Mathematische Annalen, 293 (1992), pp.~399--426.

\bibitem{Bre01}
{\sc P.~Bremaud}, {\em Markov Chains: Gibbs Fields, Monte Carlo Simulation, and
  Queues}, Texts in Applied Mathematics, Springer New York, 2001.

\bibitem{DaYo18}
{\sc S.~Das and J.~Yorke}, {\em Super convergence of ergodic averages for
  quasiperiodic orbits}, Nonlinearity, 31 (2018), pp.~491--501.

\bibitem{DiSt91}
{\sc P.~Diaconis and D.~Stroock}, {\em {Geometric Bounds for Eigenvalues of
  Markov Chains}}, The Annals of Applied Probability, 1 (1991), pp.~36 -- 61.

\bibitem{DuSc02}
{\sc F.~Durand and D.~Schneider}, {\em Ergodic averages with deterministic
  weights}, Annales de l'Institut Fourier, 52 (2002), pp.~561--583.

\bibitem{KaMc95}
{\sc L.~Karam and J.~McClellan}, {\em Complex chebyshev approximation for fir
  filter design}, IEEE Transactions on Circuits and Systems II: Analog and
  Digital Signal Processing, 42 (1995), pp.~207--216.

\bibitem{Lan52}
{\sc C.~Lanczos}, {\em Solution of systems of linear equations by minimized
  iterations1}, Journal of research of the National Bureau of Standards, 49
  (1952), p.~33.

\bibitem{LePe17}
{\sc D.~Levin and Y.~Peres}, {\em Markov Chains and Mixing Times}, American
  Mathematical Society, 2017.

\bibitem{LiWe07}
{\sc M.~Lin and M.~Weber}, {\em Weighted ergodic theorems and strong laws of
  large numbers}, Ergodic Theory and Dynamical Systems, 27 (2007), p.~511?543.

\bibitem{Man77}
{\sc T.~Manteuffel}, {\em The tchebychev iteration for nonsymmetric linear
  systems.}, Numerische Mathematik, 28 (1977), pp.~307--328.

\bibitem{Ser03}
{\sc S.~Nacu}, {\em Glauber dynamics on the cycle is monotone}, Probability
  Theory and Related Fields, 127 (2003), pp.~177--185.

\bibitem{Ort22}
{\sc A.~Ortega}, {\em Introduction to Graph Signal Processing}, Cambridge
  University Press, 2022.

\bibitem{Riv81}
{\sc T.~Rivlin}, {\em An Introduction to the Approximation of Functions}, Dover
  Publications, 1981.

\bibitem{RuBi24}
{\sc M.~Ruth and D.~Bindel}, {\em Finding birkhoff averages via adaptive
  filtering}, Chaos: An Interdisciplinary Journal of Nonlinear Science, 34
  (2024), p.~123109.

\bibitem{Saa11}
{\sc Y.~Saad}, {\em Numerical Methods for Large Eigenvalue Problems}, Society
  for Industrial and Applied Mathematics, 2011.

\bibitem{SaVo00}
{\sc Y.~Saad and H.~A. {van der Vorst}}, {\em Iterative solution of linear
  systems in the 20th century}, Journal of Computational and Applied
  Mathematics, 123 (2000), pp.~1--33.

\bibitem{ToYo22}
{\sc Z.~Tong and Y.~Li}, {\em Exponential convergence of weighted birkhoff
  average}, 2022, \url{https://arxiv.org/abs/2205.09496}.

\bibitem{Tot05}
{\sc V.~Totik}, {\em Orthogonal polynomials.}, Surveys in Approximation Theory,
  1 (2005), pp.~70--125.

\bibitem{GoVa61}
{\sc G.~G. VARGA, R.S.}, {\em Chebyshev semi-iterative methods, successive
  overrelaxation iterative methods, and second order richardson iterative
  methods. part i.}, Numerische Mathematik, 3 (1961), pp.~147--156.

\bibitem{VeKoGo14}
{\sc M.~Vetterli, J.~Kovacevic, and V.~K. Goyal}, {\em Foundations of Signal
  Processing}, Cambridge University Press, 2014.

\end{thebibliography}

\end{document}